\documentclass[acmtog]{acmart}
\acmSubmissionID{437}

\usepackage[utf8]{inputenc}

\usepackage[final]{pdfpages}

\usepackage[english]{babel}

\usepackage{graphicx}

\usepackage{tikz}

\newtheorem{remark}{Remark}

\newtheorem{theorem}{Theorem}
\newtheorem{lemma}{Lemma}

\DeclareMathOperator*{\argmin}{arg\,min}
\DeclareMathOperator{\vol}{vol}
\DeclareMathOperator{\tr}{tr}
\DeclareMathOperator{\p}{\partial}
\DeclareMathOperator{\rank}{rank}
\DeclareMathOperator{\R}{\mathbb R}
\newcommand{\eps}{\varepsilon}

\usepackage{caption}

\usepackage[ruled,linesnumbered]{algorithm2e} 

\SetAlFnt{\small}
\SetAlCapFnt{\small}
\SetAlCapNameFnt{\small}
\SetAlCapHSkip{0pt}
\usepackage{xcolor, colortbl}
\definecolor{LightGreen}{rgb}{0.7,1,.7}
\definecolor{LightCyan}{rgb}{0.88,1,1}
\definecolor{Gray}{gray}{0.9}
\definecolor{DarkGreen}{rgb}{0.1,0.7,.1}
\definecolor{RED}{rgb}{1,0,0}
\newcommand\TODO[1]{{\color{red} TODO: #1 }}

\newcommand\algocomment[1]{{\color{gray}  // \emph{ #1}}}

\acmJournal{TOG}

\begin{document}
\title{Foldover-free maps in 50 lines of code}

\author{Vladimir Garanzha}
\authornote{This work is supported by the Ministry of Science and Higher Education of the Russian Federation, project No 075-15-2020-799}
\author{Igor Kaporin}
\authornotemark[1]
\author{Liudmila Kudryavtseva}
\authornotemark[1]

\affiliation{\institution{Dorodnicyn Computing Center FRC CSC RAS}, Moscow, Russia, \institution{Moscow Institute of Physics and Technology}, Moscow, Russia}
\author{Fran{\c{c}}ois Protais}
\authornote{Corresponding author: francois.protais@inria.fr}

\author{Nicolas Ray}
\author{Dmitry Sokolov}
\affiliation{\institution{Université de Lorraine}, \institution{CNRS}, \institution{Inria}, \institution{LORIA}, F-54000 Nancy, France}

\begin{abstract}

Mapping a triangulated surface to 2D space (or a tetrahedral mesh to 3D space) is the most fundamental problem in geometry processing.
In computational physics, untangling plays an important role in mesh generation: it takes a mesh as an input, and moves the vertices to get rid of foldovers.
In fact, mesh untangling can be considered as a special case of mapping where the geometry of the object is to be defined in the map space and the geometric domain is not explicit, supposing that each element is regular.
In this paper, we propose a mapping method inspired by the untangling problem and compare its performance to the state of the art.
The main advantage of our method is that the untangling aims at producing locally injective maps, which is the major challenge of mapping.
In practice, our method produces locally injective maps in very difficult settings, and with less distortion than the previous work, both in 2D and 3D. We demonstrate it on a large reference database as well as on more difficult stress tests.
For a better reproducibility, we publish the code in Python for a basic evaluation, and in C++ for more advanced applications.

\end{abstract}

%
%


%
%


\begin{teaserfigure}
\centerline{
    \includegraphics[width=\textwidth]{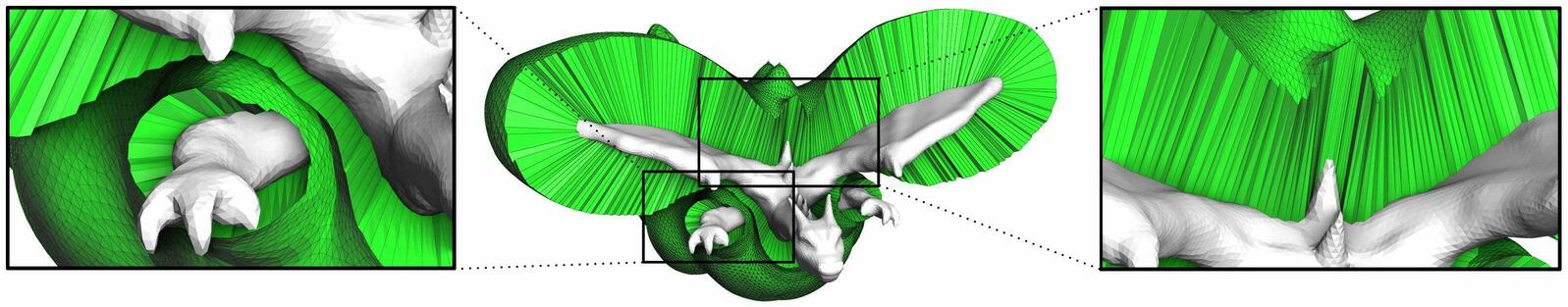}
}
\caption{
Our method of constructing injective maps opens a door for a large variety of applications.
This figure shows an example of a thick prismatic mesh layer (shown in green) built around a triangulated surface, a very challenging problem for highly curved objects.
Thanks to our method, we are able to compute such a layer free of folds and self-intersections.
}
\label{fig:teaser}
\end{teaserfigure}

\maketitle

\section{Introduction}

Most geometric objects are represented by a triangulated surface or a tetrahedral mesh. The mapping problem consists in generating a 2D or 3D map of these objects.
This is a fundamental problem of computer graphics because it is much easier for many applications to work in this map space than to directly manipulate the object itself.
To give few examples, texture mapping stores colors of a surface as images in the map space, remeshing uses global maps in 2D \cite{Bommes2013} and 3D \cite{Gregson2011,Nieser2011}. In addition, mapping algorithms can be used to deform volumes \cite{Li2020}, or generate shells from surfaces \cite{Jiang2020}.

\paragraph*{What is a good map?}
Most often maps are represented by the position of the vertices in the map space, and interpolated linearly on each element (triangle or tetrahedron).
In a perfect world, the map space would keep the geodesic distances of the object. Unfortunately, this is usually impossible due to Gaussian curvature, and application-specific constraints such as constrained position of vertices or overlaps in the map space. Therefore, the objective of the mapping algorithms is to minimize the distortion between geometric and map spaces, opening the door for numerical optimization approaches as detailed in the survey \cite{Hormann2008}.

\paragraph*{What about the invertibility?}
Unfortunately, when high distortion is required to satisfy the constraints, these algorithms may lose the fundamental property of a maps: injectivity.
A solution to preserve it \cite{Floater97} relies on Tutte's theorem \cite{Tutte1963}, however the surface boundary must be mapped to a convex polygon.
Despite this strong limitation, it still remains the reference algorithm to generate injective maps. Lower distortion can be obtained by changing weights of the barycentric coordinates \cite{Eck1995} (as long as they are not negative), and alternative solutions \cite{Campen2016,Shen2019} have been explored to improve robustness to numerical imprecision by modifying the mesh connectivity.

\paragraph*{Local invertibility.}
In many applications, maps are used to access a neighborhood of a point within a coherent local coordinate system.
To this end, global injectivity is not required, and we instead look for local injectivity \cite{Schuller2013,Smith2015}.
Their approach starts from an injective map \cite{Floater97}, and maintains the local injectivity when minimizing the distortion and enforcing the constraints.
This allows them to optimize at the same time the parameterization and the texture packing \cite{jiang2017}, with a possibility to scale to larger meshes \cite{Rabinovich2017}.

\paragraph*{Recover local injectivity.}
Local injectivity can also be recovered for a map with few foldovers present.
For example, in 2D \cite{Lipman2012} and 3D \cite{Aigerman2013}, the map is projected on a class of bounded distortion maps.
The numerical method is however unlikely to succeed for stiff problems.
Recovering local injectivity is also known as mesh untangling.
Originally related to Arbitrary Lagrangian-Eulerian moving mesh approach, the mesh untangling problem considers a simplicial complex with misoriented elements and attempts to flip them back by optimizing the position of the vertices.
There is an abundant literature on mesh untangling \cite{Du2020,knupp2001hexahedral,freitag2000local,escobar2003simultaneous,toulorge2013robust},
however the common opinion is that untangling is a very hard problem and algorithms are not robust enough.
As a manifestation of frustration over this problem \cite{Danczyk2013} investigates a finite element method working directly on tangled (\emph{sic!}) meshes.



\paragraph*{Elastic deformations.}
To recover local injectivity, we propose a method stemming from the computational physics.
It is very important to note that there is rich literature on mesh deformation in the community working on grid generation for scientific computation.
Numerical simulation of hydrodynamic instability of layered structures requires sound mathematical foundations behind moving deforming mesh algorithms.
In the '60s Winslow and Crowley, independently one from another, introduced mesh generation methods based on inverse harmonic maps \cite{crowley1962equipotential,winslow1966numerical}.

Since then, a lot of effort was spent on mesh generation based on elastic deformations~\cite{jacquotte1988mechanical}, but mostly for regular grids.
In 1988, at the time of domination of finite difference mapped grid generation methods, S. Ivanenko introduced the pioneering concept of barrier variational grid generations methods guaranteeing construction of non-degenerate grids~\cite{ivanenko1988construction,charakhch1997variational}.
To generate deformations with bounded global distortion (bounded quasi-isometry constant), Garanzha proposed to minimize an elastic energy for a hyperelastic material with
stiffening suppressing singular deformations~\cite{garanzha2000barrier}. Invertibility theorem for deformation of this material was established in the 3D case as well~\cite{Garanzha2014}.

A solid mathematical ground for these methods was laid by J. Ball who introduced in 1976 his theory of finite elasticity based on the concept of polyconvex distortion energies \cite{ball1976convexity}.
He not only proved Weierstrass-style existence theorem for this class of variational problems, but also formulated a theorem on invertibility of elastic deformations for quite general 3D domains \cite{Ball1981}.
It is important to note that Ball invertibility theorem is proved for Sobolev mappings and can be applied directly for finite element spaces, i.e. to deformation of meshes, as was pointed out in \cite{rumpf1996variational}.




\paragraph{Our contributions}
Inspired by these results on untangling and elastic deformations, we propose a method outperforming recent state of the art on locally injective parameterization \cite{Du2020} in
terms of robustness, quality and supported features. To sum up, our contributions are:
\begin{itemize}
 \item we propose a method that optimizes the map distortion directly, with a parameter for a tradeoff between the angle and area preservation,
 \item we provide a rigourous analysis of the foundations of our numerical resolution scheme,
 \item our method supports mapping with free boundaries,
 \item we observe better robustness for high distortions as well as with respect to a poor quality initialization.
 \item To ease the reproducibility, we publish the code in Python (refer to Listing~1) for a basic evaluation, and a C++ code in the supplemental material for more advanced applications.
\end{itemize}

The rest of the paper is organized as follows: we start with presenting our method in \S~\ref{sec:penatly}, then we evaluate its performance (\S~\ref{sec:results:benchmark} and \S~\ref{sec:results:testing}) as well as its limitations (\S~\ref{sec:results:limitations}).
Then we present theoretical guarantees for our resolution scheme:
in \S~\ref{sec:hessian} we prove that our approximation of Hessian matrix is definite positive,
and finally in \S~\ref{sec:finite-untangling} we prove that our choice of the regularization parameter sequence guarantees that a minimization algorithm\footnote{The minimization algorithm is subject to conditions of Th.~\eqref{th:th}; from our numerical experiments we observe that our minimization algorithm almost always respects the conditions.} can find a mesh free of inverted elements in a finite number of steps.

\section{Penalty method for mesh untangling}
\label{sec:penatly}


In this rather short section we present our method of computing a foldover-free map $\vec u : \Omega\subset \mathbb R^d \rightarrow \mathbb R^d$, i.e. we map the domain $\Omega$ to a parametric domain.
This presentation is unified both for 2D and 3D settings, and by $d$ we denote the number of dimensions; in our notations we use arrows for all vectors of dimension $d$.

The section is organized as follows: in \S~\ref{sec:penalty:intro} we give a primer on the variational formulation of mapping problem in continuous settings,
then we state our problem in \S~\ref{sec:penalty:pb} as a regularization of this variational formulation, and finally we present our numerical resolution scheme in \S~\ref{sec:penalty:resolution}.

\subsection{Variational formulation for grid generation}
\label{sec:penalty:intro}

Let us denote by $\vec u(\vec x)$ a map to a parametric domain: for the flat 2D case we can write as $\vec u (x,y) = (u(x,y), v(x,y))$, and for a 3D map $\vec u (x,y,z) = (u(x,y,z), v(x,y,z), w(x,y,z))$.

Consider the following variational problem:
\begin{equation}
\argmin\limits_{\vec{u}}\int\limits_\Omega\left( f(J) + \lambda g(J) \right)\,dx,
\label{eq:winslow}
\end{equation}
where $J$ is the Jacobian matrix of the mapping $\vec u (\vec x)$, and
\[
f(J) := \left\{ \begin{array}{ll} \frac{\tr J^\top J}{(\det J)^\frac2d}, & \det J > 0 \\
+\infty, & \det J \leq 0 \end{array} \right.
\]

\[
g(J) := \left\{ \begin{array}{ll} \det J + \frac{1}{\det J}, & \det J > 0 \\
+\infty, & \det J \leq 0 \end{array} \right.
\]
Problem~\eqref{eq:winslow} may be subject to some constraints that we do not write explicitly. To give an example, one may pin some points in the map.
In this formulation, functions $f(J)$ and $g(J)$ have concurrent goals, one preserves angles and the other preserves the area, and thus $\lambda$ serves as a trade-off parameter.

As a side note, with $\lambda = 0$ and $d=2$, Prob.~\eqref{eq:winslow} presents a variational formulation of an inverse harmonic map problem.
Namely, if we write down the Euler-Lagrange equations for Prob.~\eqref{eq:winslow} and interchange the dependent and independent variables\footnote{Attention, this step assumes that the solution of Prob.~\eqref{eq:winslow} is a diffeomorphism.}, we obtain the Laplace equation $\Delta \vec{x}(\vec{u}) = \vec{0}$ (not to be confused with omnipresent $\Delta\vec{u}(\vec x)=\vec{0}$!).
For this case, Prob.~\eqref{eq:winslow} is often reffered as to Winslow's functional, however Winslow himself has never formulated the variational problem, working with inverse Laplace equations.
To the best of our knowledge, the first publication of the variational problem is made by Brackbill and Saltzman~\cite{Brackbill1982}.

Prob.~\eqref{eq:winslow} provides a very powerful tool based on the theory of finite hyperelasticity introduced by J. Ball~\cite{ball1976convexity}.
From a computational point of view, starting with a map without foldovers, this tool allows us to optimize the quality of the map.
Note that the energy~\eqref{eq:winslow} is a polyconvex function (refer to App.~\ref{app:positive}--Rem.~\ref{remark:polyconvexity} for a proof) satisfying the ellipticity conditions,
and therefore is very well suited for a numerical optimization provided that we have an initial guess in the admissible domain $\min\limits_\Omega J(\vec u)>0$.

The problem, however, is that while being theoretically sound, this problem statement does not offer any practical way to get rid of foldovers in a map,
because for a map with foldovers the energy is infinite and provides no indications on how to improve the situation,
hence we propose to alter a little the problem statement.

\subsection{Penalty method}
\label{sec:penalty:pb}
We can avoid nonpositive denominators in $f$ and $g$ using a regularization function $\chi$ for a positive value of $\varepsilon$ (Fig.~\ref{fig:chi}):
\begin{equation}
\chi(D, \eps) := \frac{D+\sqrt{\eps^2 + D^2}}{2},
\label{eq:chi}
\end{equation}

\begin{figure}[!t]
\centering
\includegraphics[width=.5\linewidth]{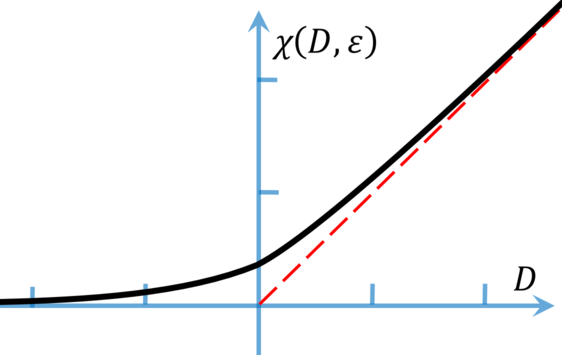}
\caption{Regularization function for the denominator in Eq.~\eqref{eq:fege}. When $\varepsilon$ tends to zero, $\chi(\varepsilon, D)$ tends to $D$ for positive values of $D$, and to $0^+$ for negative values of $D$.}
\label{fig:chi}
\end{figure}

\noindent Then we define a regularized version $f_\eps, g_\eps$ of functions $f$ and $g$:
\begin{equation}
f_\varepsilon(J) := \frac{\tr J^\top J}{\left(\chi(\det J, \varepsilon)\right)^\frac2d},
\qquad
g_\varepsilon(J) := \frac{\det^2 J + 1}{\chi(\det J, \varepsilon)},
\label{eq:fege}
\end{equation}
\noindent
so that Prob~\eqref{eq:winslow} is reformulared as 
\begin{equation}
\lim\limits_{\varepsilon \rightarrow 0^+}\argmin\limits_{\vec{u}}\int\limits_\Omega\left( f_\varepsilon(J) + \lambda g_\varepsilon(J) \right) \,dx
\label{eq:continuous}
\end{equation}

Under certain assumptions\footnote{The fact that the solution of Prob.~\eqref{eq:continuous} is a diffeomorphism is sufficient (but not necessary) for the equivalence.}
solutions of Prob.~\eqref{eq:continuous} are solutions of Prob.~\eqref{eq:winslow}, however, Prob.~\eqref{eq:continuous} does offer a way of getting rid of foldovers if a foldover-free initialization is not available.

In practice, the map $\vec u$ is piecewise affine with the Jacobian matrix $J$ being piecewise constant, and can be represented by the coordinates of the vertices in the parametric domain $\{\vec{u}_i\}_{i=1}^{\#V}$.
Let us denote the vector of all variables as $U := \left(\vec{u}_1^\top \dots \vec{u}_{\#V}^\top \right)^\top$, then our optimization problem can be discretized as follows:

\begin{gather}
\label{eq:discrete}
\lim\limits_{\varepsilon \rightarrow 0^+}\argmin\limits_{U} F(U,\varepsilon), \\
\text{ where }\quad  F(U, \varepsilon) : =\sum\limits_{t=1}^{\#T} \left(f_\varepsilon(J_t) + \lambda g_\varepsilon(J_t)\right)\vol(T_t) \nonumber,
\end{gather}
$\#V$ is the number of vertices, $\#T$ is the number of simplices, $J_t$ is the Jacobian matrix for the simplex $t$ and $\vol(T_t)$ is the volume of the simplex $T_t$ in the original domain.

\subsection{Resolution scheme}
\label{sec:penalty:resolution}

To solve Prob.~\eqref{eq:discrete}, we use an iterative descent method.
We start from an initial guess $U^0$, and we build a sequence of approximations $U^{k+1} := U^{k} + \Delta U^{k}$.
For each iteration we need to carefully choose the regularization parameter $\varepsilon^k$.
Starting from $\varepsilon^0 := 1$, we define the sequence as follows:
\begin{equation}
\label{eq:epsilon}
\eps^{k+1} :=  \left\{
\begin{array}{lcl}
\left(1 - \frac{\sigma^k\sqrt{(D_-^{k+1})^2+(\eps^k)^2}}
{|D_-^{k+1}|+\sqrt{(D_-^{k+1})^2+(\eps^k)^2}}\right)\eps^{k}, 
&\mbox{ if }& D_-^{k+1} < 0, \\ 
&&\\
(1-\sigma^k)\eps^{k}, 
&\mbox{ if }& D_-^{k+1} \ge 0, \\
\end{array} 
\right.
\end{equation}
where $D_-^{k+1}:=\min\limits_{t\in 1\dots\#T} \det J_t^{k+1}$ is the minimum value of the Jacobian determinant over all cells of the mesh at the iteration $k+1$,
and $\sigma^k := \max\left(\frac1{10}, 1 -  \frac{F(U^{k+1}, \eps^k)}{F(U^{k}, \eps^k)}\right)$.
Note that Eq.~\eqref{eq:epsilon} is justified by Th.~\ref{th:th} (\S~\ref{sec:finite-untangling}) on finite untangling sequence; while this formula provides some guarantees,
one may use a simpler heuristic to accelerate the computations, more on this in \S~\ref{sec:results:testing}.

The simplest way to find $\Delta U^{k}$ is to call a quasi-Newtonian solver such as L-BFGS-B~\cite{LBFGS}.
The only thing we need to implement is the computation of the function $F(U^k, \varepsilon^k)$ and its gradient $\nabla F(U^k, \varepsilon^k)$.

Another option is to compute analytically the Hessian matrix instead of estimating it. The problem, however, is that the Hessian matrix $\frac{\p^2 F}{\p U \p U^\top}$ is not positive definite.
In this paper we propose its approximation that ensures the positive definiteness.
The modified Hessian matrix $H^+(U^k,\varepsilon^k)$ of the function $F$ with respect to $U$ at the point $U^k$ is built out of $d\times d$ blocks
$$
H^+_{ij} \approx \frac{\p^2 F}{\p \vec{u}_i \p \vec{u}_j^\top}(U^k,\varepsilon^k).
$$
Here, the matrix $H^+_{ij}$ is placed on the intersection of $i$-th block row and $j$-th block column; the $\approx$ symbol means that we remove all the terms depending on the second derivative of $\chi$
and second derivatives of $\det J$ to keep $H^+$ positive definite.
Refer to Appendix~\ref{app:gradient-hessian} for the formulæ, and to \S~\ref{sec:hessian} for the proof of the positive definiteness of $H+$.

A detailed description of the resolution scheme is given in Alg.~\ref{alg:lbfgs}.
Refer to List.~1 and Fig.~\ref{fig:belinsky-z} for a complete working example of Python implementation and the corresponding input/output generated by the code.
Note that our method is not limited to simplicial meshes only: in this particular example we evaluate the Jacobian matrix for every triangle forming quad corners, what corresponds to the trapezoidal quadrature rule.

\begin{algorithm}[!t]
\KwIn{$U^{0}$; \algocomment{initial guess (vector of size $\#V \times d$)} }
\KwIn{\textit{useQuasiNewton}; \algocomment{boolean to choose the optimization scheme} }
\KwOut{$U$; \algocomment{final locally injective map (vector of size $\#V \times d$)} }
$k\leftarrow 0$\;
\Repeat{
\footnotesize
$\min\limits_{t\in 1\dots\#T} \det J_t^{k}>0$ \textbf{~and~} $F(U^{k}, \varepsilon^{k})>(1-10^{-3})\, F(U^{k-1}, \varepsilon^{k-1})$ 
}{
    compute $\varepsilon^k$ ; \algocomment{regularization parameter, Eq.~\eqref{eq:epsilon} [variant: Eq.~\eqref{eq:epsilon:heuristic}]}\\
    \eIf{useQuasiNewton}{
    $U^{k+1} \leftarrow \text{ L-BFGS-B }(U^k, \varepsilon^k)$; \algocomment{inner L-BFGS-B loop}\\
    }{
        compute a modified Hessian matrix $H^+(U^k,\varepsilon^k)$; 
        $\Delta U^k \leftarrow (H^+)^{-1} \ \nabla F(U^k, \varepsilon^k)$; \algocomment{conjugate gradients}\\
        $U^{k+1} \leftarrow \argmin\limits_{\tau}F(U^k + \tau \, \Delta U^{k}, \varepsilon^k)$; \algocomment{line search}\\
    }
    $k \leftarrow k+1$\;
}
$U\leftarrow U^k$\;
\caption{Computation of a locally injective map}
\label{alg:lbfgs}
\end{algorithm}

\begin{figure}[!t]
\centering
\includegraphics[width=.48\linewidth]{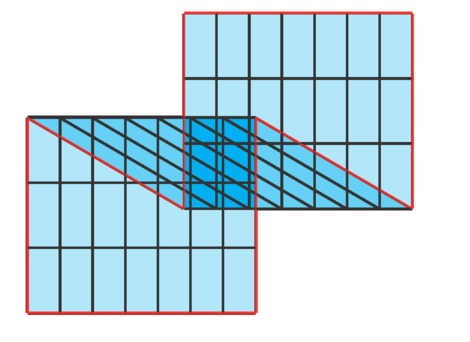}
\includegraphics[width=.48\linewidth]{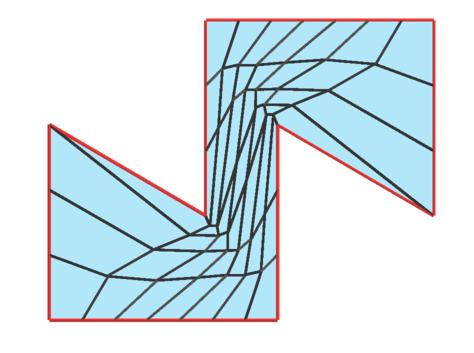}
\caption{
The input mesh with foldovers and the untangling produced by the Listings~1 and~2.
\textbf{Left:} a quad 2D mesh to untangle. The boundary (in red) is locked, and the black mesh is free to move. \textbf{Right:} fold-free result.
}
\label{fig:belinsky-z}
\end{figure}

\begin{figure}[tb]
\centering
\includegraphics[width=.6\linewidth]{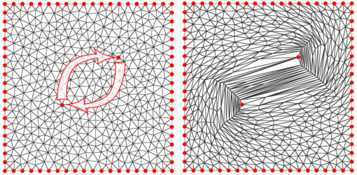}
\caption{
Injective mapping sanity check: make two points inside a square switch places. Left: the input problem, all locked points are shown in red. Right: foldover-free result obtained with our method.
}
\label{fig:sanity-check}
\end{figure}

\begin{figure*}[!p]
\centering
\begin{tikzpicture}
\draw (0, 0) node[inner sep=0] {\includegraphics[width=\linewidth]{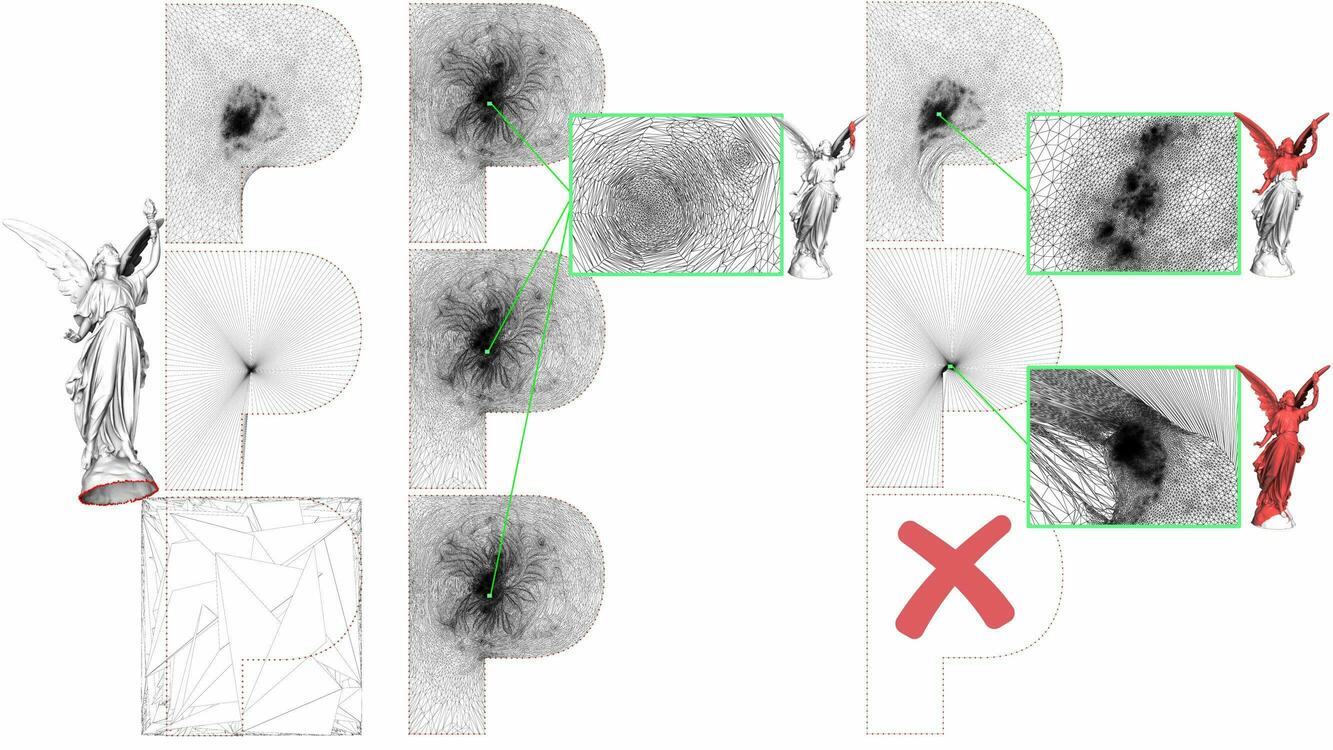}};
\draw (-5.5, 5.3) node {\LARGE Initialization};
\draw (-0.7, 5.3) node {\LARGE Our result};
\draw (5.7, 5.3) node {\LARGE \cite{Du2020}};
\draw (0.85, 1) node {min det $J\approx 4\cdot 10^{-4}$}; 
\draw (0.85, 0.5) node {max $\frac{\sigma_1(J)}{\sigma_2(J)}\approx 69$};
\draw (0.85, 4.5) node {min det $J\approx 4\cdot 10^{-4}$}; 
\draw (0.85, 4) node {max $\frac{\sigma_1(J)}{\sigma_2(J)}\approx 72$};
\draw (0.85, -2) node {min det $J\approx 4\cdot 10^{-4}$}; 
\draw (0.85, -2.5) node {max $\frac{\sigma_1(J)}{\sigma_2(J)}\approx 70$};
\draw (7.5, 1) node {min det $J\approx 1 \cdot 10^{-20}$}; 
\draw (7.5, 0.5) node {max $\frac{\sigma_1(J)}{\sigma_2(J)}\approx 24104$};
\draw (7.5, 4.5) node {min det $J\approx 9\cdot 10^{-17}$}; 
\draw (7.5, 4) node {max $\frac{\sigma_1(J)}{\sigma_2(J)}\approx 16500$};
\end{tikzpicture}
\caption{
Constrained boundary injective mapping challenge proposed by~\cite{Du2020}:
the ``Lucy'' mesh is mapped to a P-shaped domain by constraining the vertices shown in red.
Left column: three different initializations for the same problem.
Middle column: our method produces the same (up to a numerical precision) result on all three initializations.
Right column: total lifted content method~\cite{Du2020} fails to solve for the randomly initialized interior vertices, and produces very different results on other two initializations.
Three miniature images of ``Lucy'' show in red the portion of the surface visible in the corresponding close-ups.
}
\label{fig:lucy}

\vspace{3mm}

\begin{minipage}{.15\linewidth}
\centering
\includegraphics[width=\linewidth]{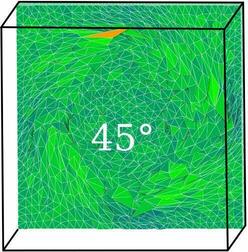}
\includegraphics[width=\linewidth]{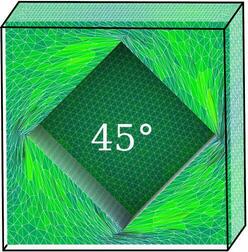}
$
\max\frac{\sigma_1(J)}{\sigma_3(J)} \approx 37
$\\
$
\min\det J \approx 0.03
$\\
\textbf{(a)}
\end{minipage}
\begin{minipage}{.15\linewidth}
\centering
\includegraphics[width=\linewidth]{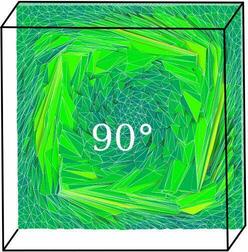}
\includegraphics[width=\linewidth]{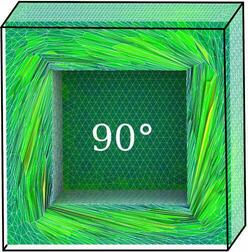}
$
\max\frac{\sigma_1(J)}{\sigma_3(J)} \approx 58
$\\
$
\min\det J \approx 0.02
$\\
\textbf{(b)}
\end{minipage}
\begin{minipage}{.15\linewidth}
\centering
\includegraphics[width=\linewidth]{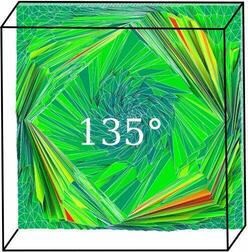}
\includegraphics[width=\linewidth]{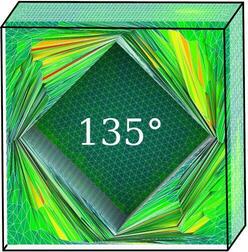}
$
\max\frac{\sigma_1(J)}{\sigma_3(J)} \approx 759
$\\
$
\min\det J\approx 0.003
$\\
\textbf{(c)}
\end{minipage}
\begin{minipage}{.15\linewidth}
\centering
\includegraphics[width=\linewidth]{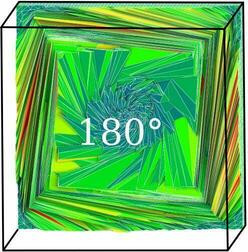}
\includegraphics[width=\linewidth]{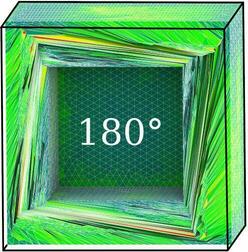}
$
\max\frac{\sigma_1(J)}{\sigma_3(J)} \approx 659
$\\
$
\min\det J\approx 0.002
$\\
\textbf{(d)}
\end{minipage}
\hspace{6mm}
\begin{minipage}{.15\linewidth}
\centering
\includegraphics[width=\linewidth]{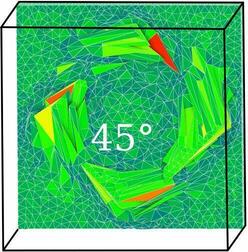}
\includegraphics[width=\linewidth]{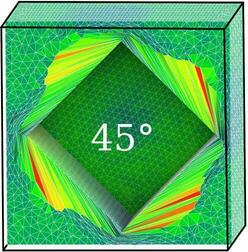}
$
\max\frac{\sigma_1(J)}{\sigma_3(J)} \approx 524
$\\
$
\min\det J\approx 0.002
$\\
\textbf{(e)}
\end{minipage}
\begin{minipage}{.15\linewidth}
\centering
\includegraphics[width=\linewidth]{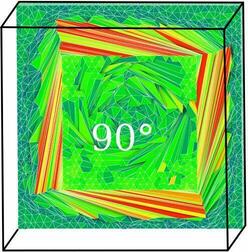}
\includegraphics[width=\linewidth]{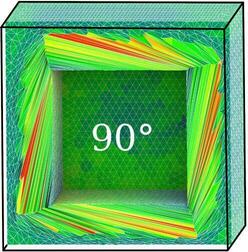}
$
\max\frac{\sigma_1(J)}{\sigma_3(J)} \approx 3391
$\\
$
\min\det J\approx 0.0004
$\\
\textbf{(f)}
\end{minipage}

\caption{Constrained boundary injective mapping stress test. We have generated an isotropic tetrahedral mesh of a cube subtracted from a larger cube,
and we rotate the inner cube's boundary to test the robustness.
Top row and bottom row correspond to two different slices of the same mesh.
Columns \textbf{(a)--(d)}: injective maps produced by our method, columns \textbf{(e)} and \textbf{(f)}: injective maps produced by~\cite{Du2020}.
The method by~\cite{Du2020} fails to generate injective maps for 135° and 180° inner cube rotations.
The colormap illustrates the relative volume scaling: green for $\det J\approx 1$, red for inflation, blue for compression.
}
\label{fig:stress1}
\end{figure*}

\section{Results and discussion}
\label{sec:results}

In this section we provide an experimental evaluation of the method.
In the field of computer graphics, any claim about map injectivity always faces a simple sanity check (Fig.~\ref{fig:sanity-check}):
take a square and make two points inside switch places. Our method successfully avoids the desk-reject, so
we start this section (\S~\ref{sec:results:benchmark}) by testing our method on the benchmark~\cite{Du2020},
then we continue with further tests we have found relevant (\S~\ref{sec:results:testing}), and finally we discuss the limitations of the approach in \S~\ref{sec:results:limitations}.

\subsection{Benchmark database}
\label{sec:results:benchmark}
Along with their paper, Du et al. have published a valuable benchmark database.
It contains a huge number of 2D and 3D constrained boundary injective mapping challenges.
For 2D challenges, the benchmark contains 3D triangulated surfaces to flatten, and not flat 2D meshes as we have described in \S~\ref{sec:penalty:intro}.
Nevertheless, our method can handle it directly because the mapping is still $\mathbb R^2 \rightarrow \mathbb R^2$ on each triangle.
To the best of our knowledge, Total Lifted Content (TLC)~\cite{Du2020} and our method are the only ones passing the benchmark without any fail.

A representative example from the database is given in top row of Fig.~\ref{fig:lucy}.
The challenge is to map the ``Lucy'' mesh statuette from the Stanford Computer Graphics Laboratory to a P-shaped domain.
This mesh has a topology of a disk, and its boundary vertices are uniformly spaced on the P-shape boundary.
As an initialization to the problem, Du et al. have computed the corresponding (flat) minimal surface that obviously contains a foldover (Fig.~\ref{fig:lucy}--top left).
Then the problem boils down to a mesh untangling with locked boundary.

\paragraph*{Mapping quality measure}
How to measure quality of a map? Well, it depends on the goal. An identity is an unreachable ideal; traditional competing goals are (as much as possible) angle preserving and area preserving maps.
Thus, we can measure the extreme values of the failure of a map to be conformal or authalic.
Our maps being piecewise affine, the Jacobian matrix $J$ is constant per element.
Let us define the largest singular value of $J$ as $\sigma_1(J)$, and the smallest singular value as $\sigma_d(J)$;
then the quality of a mapping can be reduced to extreme values of the stretch ($\max\frac{\sigma_1(J)}{\sigma_d(J)}$) and the scaling ($\min \det J$).

For the ``Lucy-to-P'' challenge (Fig.~\ref{fig:lucy}--top row) our map differs from the TLC result by 12 orders of magnitude in terms of minimum scaling, and by two orders of magnitude in terms of maximum stretch.
To visualize this difference in scaling, we have provided the close-ups: Fig.~\ref{fig:lucy}--top middle shows a map of the Lucy's torch, whereas the same level of zoom on the result by Du et al. (Fig.~\ref{fig:lucy}--top right)
contains not only the torch, but also both wings, the head and the right arm!

Note also that the input ``Lucy'' mesh is slightly anisotropic; our method allows us to prescribe the element target shape, so the dress pleats are clearly visible in our mapping.

\begin{figure}[!t]
\centering
\textbf{2D dataset (10743 challenges)} \\
\includegraphics[width=.48\linewidth]{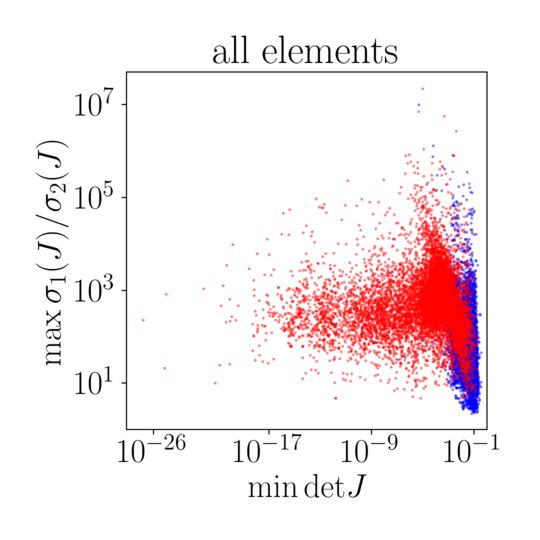}
\includegraphics[width=.48\linewidth]{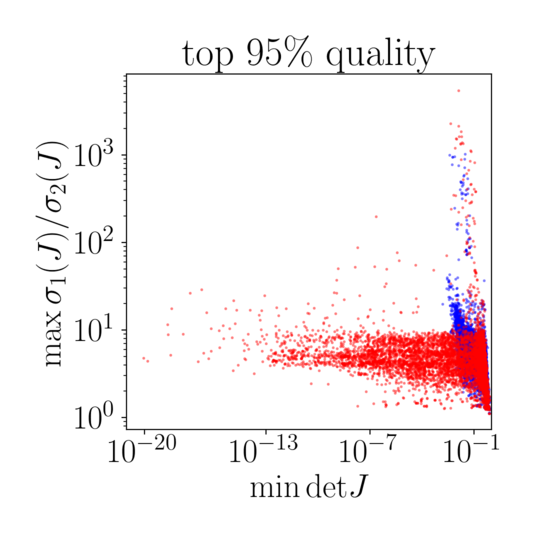}\\
\textbf{3D dataset (904 challenges)} \\
\includegraphics[width=.48\linewidth]{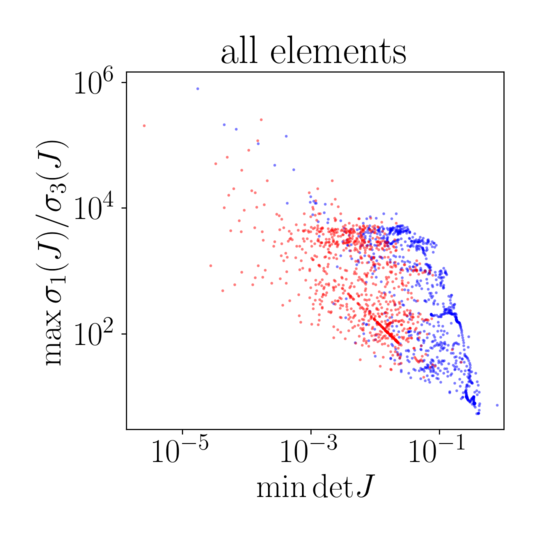}
\includegraphics[width=.48\linewidth]{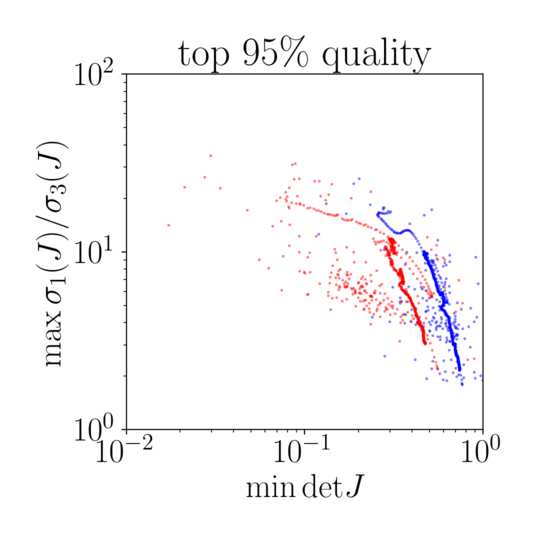}
\caption{Quality plot of the resulting locally injective maps for every challenge from the database provided by~\cite{Du2020}.
Our results are shown in blue, whereas the results by Du et al. are shown in red.
Each dot corresponds to a quality of the corresponding map reduced to two numbers: the maximum stretch and the mininum scale.
\textbf{Top row:} mapping quality on the 2D dataset. \textbf{Bottom row:} mapping quality on the 3D dataset.
Left column shows the absolute maximum stretch and absolute minimum scale, whereas the right column shows the maximum stretch and minimum scale for the top 95\% of measurements.
}
\label{fig:scatter}
\end{figure}

\paragraph*{Benchmark database}
Our method successfully passes all challenges from the benchmark~\cite{Du2020}.
The benchmark consists of 10743 meshes to untangle in 2D and 904 meshes in 3D under locked boundary constraints.
In Fig.~\ref{fig:scatter} we provide quality plots of the resulting locally injective maps.
These are $\log-\log$ scatter plots: each dot corresponds to a quality of the corresponding map reduced to two numbers: the maximum stretch ($\max\frac{\sigma_1(J)}{\sigma_d(J)}$) and the minimum scaling ($\min \det J$).
Left column of Fig.~\ref{fig:scatter} shows the worst quality measurements for every 2D problem (top) as well as for every 3D challenge (bottom) of the dataset.
Our results are shown in blue, whereas TLC results are shown in red.
To illustrate the distribution of the elements' quality, for each injective map we have removed 5\% of worst measurements:
the right column of Fig.~\ref{fig:scatter} shows the maximum stretch and the minimum scaling for the top 95\% of measurements.

Note the dot arrangements forming lines in the plot: these dots correspond to the few sequences of deformation present in the database.

\begin{figure}[!t]
\centering
\begin{minipage}{.48\linewidth}
\centering
\includegraphics[width=\linewidth]{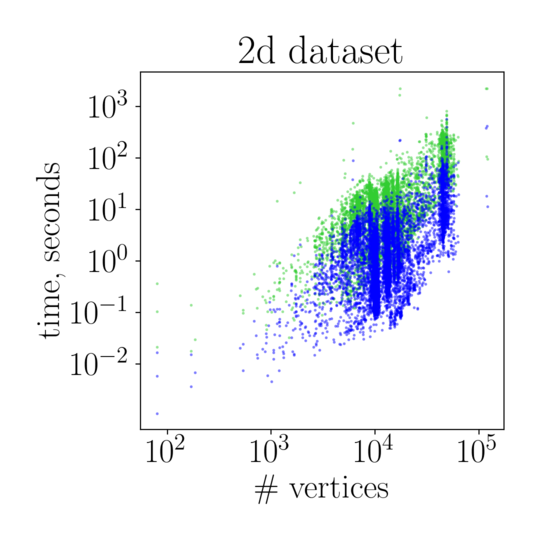}
\end{minipage}
\begin{minipage}{.48\linewidth}
\centering
\includegraphics[width=\linewidth]{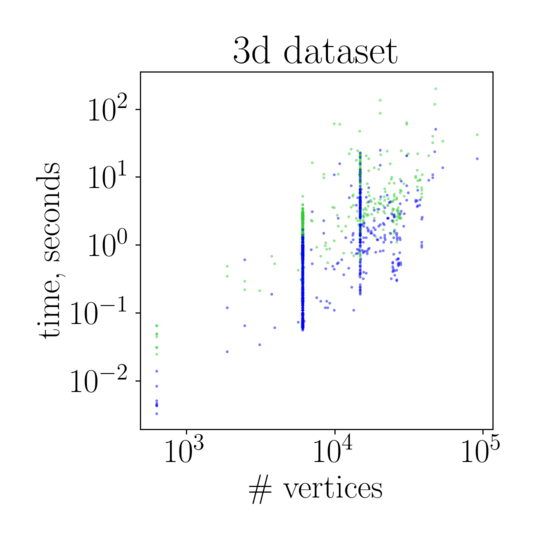}
\end{minipage}
\caption{Performance of our method tested on the benchmark \cite{Du2020}.
Each dot corresponds to a challenge from the database (10743 in 2D and 904 in 3D).
Blue dots show the running times obtained using Eq.~\eqref{eq:epsilon:heuristic}, green dots correspond to Eq.~\eqref{eq:epsilon}.
}
\label{fig:timings}
\end{figure}

\paragraph*{Timings}
Fig.~\ref{fig:timings} provides a $\log-\log$ scatter plot of our running time vs mesh size for all the challenges from the database \cite{Du2020}: for each run, the time varies from a fraction of a second to several minutes for the largest meshes.
These times were obtained with a 12 cores i7-6800K CPU @ 3.40 GHz.
As in Fig.~\ref{fig:timings}, the vertical lines in the 3D dataset plot correspond to the sequences of deformation in the benchmark.

There are two scatter plots superposed, both represent the same resolution scheme with an exception corresponding to the way we compute $\eps^k$ (Alg.~\ref{alg:lbfgs}--line 3).
The green scatter plot corresponds to a conservative update rule (Eq.~\eqref{eq:epsilon}) offering guarantees on untangling in a finite number of steps (refer to Th.~\ref{th:th}), whereas
the blue scatter plot is obtained using the following update rule:
\begin{equation}
\varepsilon^k := \sqrt{10^{-12} + 4 \cdot 10^{-2}\cdot \left[ \min(0, \min\limits_{t\in 1\dots\#T} \det J_t^k)\right]^2 },
\label{eq:epsilon:heuristic}
\end{equation}
This formula was chosen empirically, and it performs very well in the vast majority of situations. In particular, it allows for all the database~\cite{Du2020} to pass the injectivity test.

\subsection{Further testing}
\label{sec:results:testing}

\paragraph*{Sensitivity to initialization}
Our next test is the sensitivity to the initialization.
We have generated two other initializations for the ``Lucy-to-P'' challenge: the one with all interior vertices collapsed onto a single point (Fig.~\ref{fig:lucy}--middle row),
and with the interior vertices being randomly placed withing a bounding square (Fig.~\ref{fig:lucy}--bottom row).

Our method produces virtually the same result on all three initializations, whereas TLC generates very different results for the first two, and fails for the third one.
It is interesting to note that TLC is heavily depending on the initialization, it alters very little the input geometry.

\paragraph*{Large deformation stress test}
For our next test we have generated an isotropic tetrahedral mesh of a cube with a cavity, and we rotated the inner boundary to test the robustness of our method to large deformations.
Figure~\ref{fig:stress1} shows the results.
Our L-BFGS-based optimization scheme succeeds up to the rotation of 135°, and we had to switch to the Newton method to reach the 180° rotation.
TLC method had succeded on 45° and 90°, and failed for the 135° and 180°.
Note that as in the previous test, even when the untangling succeeds, TLC alters very little the input map, thus producing heavily stretched tetrahedra,
whereas our method evenly dissipates the stress over all the domain.

\begin{figure}[!t]
\centering
\begin{minipage}{.58\linewidth}
\centering
\includegraphics[width=.8\linewidth]{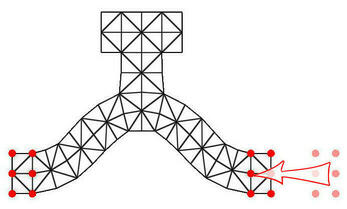}\\
\textbf{(a)}
\end{minipage}
\begin{minipage}{.38\linewidth}
\centering
\includegraphics[width=\linewidth]{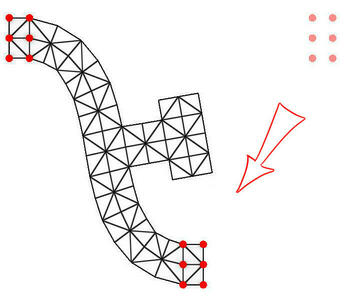}
\textbf{(b)}
\end{minipage}
\caption{Free boundary injective mapping.
The vertices shown in red are constrained, all other vertices are free to move.
\textbf{(a):} a compression test, \textbf{(b):} a bend test.
Refer to Fig~\ref{fig:lambda}--a for the rest shape.
}
\label{fig:bend-compress}
\end{figure}

\paragraph*{Free boundary injective mapping}
To the best of our knowledge, our method is the only one able to produce inversion-free maps with free boundary.
Since TLC tries to minimize the overall volume, relaxing the boundary constraints results in degenerate maps.

Fig.~\ref{fig:bend-compress} shows two maps obtained with our method: a 2D shape being compressed and the same shape being bent. The boundary is free to move, we lock the vertices shown in red.
Refer to Fig.~\ref{fig:lambda}-a for the rest shape.
The shape behaves exactly as a human would expect it: upon compression the shape chooses one of the two possible results (Fig.~\ref{fig:bend-compress}--a),
and successfully passes the bend test (Fig.~\ref{fig:bend-compress}--b), note the geometrical details that are naturally rotated.

\begin{figure}[!t]
\centering
\begin{minipage}{.48\linewidth}
\centering
\includegraphics[width=.75\linewidth]{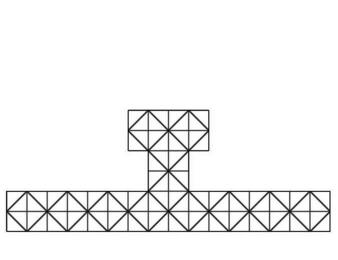}\\
\textbf{(a)}
\end{minipage}
\begin{minipage}{.48\linewidth}
\centering
\includegraphics[width=\linewidth]{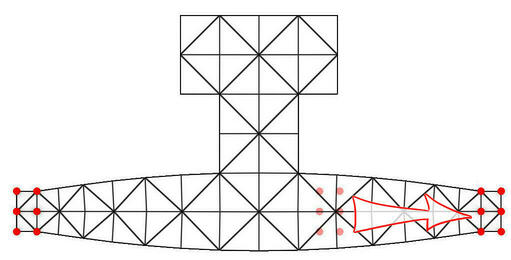}
\textbf{(b)}
\end{minipage}
\begin{minipage}{.48\linewidth}
\centering
\includegraphics[width=\linewidth]{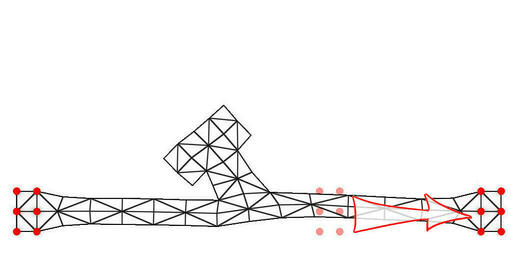}
\textbf{(c)}
\end{minipage}
\begin{minipage}{.48\linewidth}
\centering
\includegraphics[width=\linewidth]{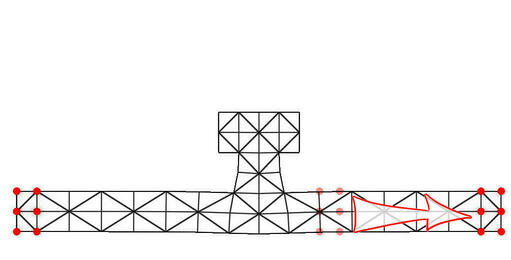}
\textbf{(d)}
\end{minipage}
\caption{Free boundary injective mapping: influence of the parameter $\lambda$ in Prob.~\eqref{eq:discrete}.
The vertices shown in red are constrained, all other vertices are free to move.
\textbf{(a):} the rest shape; \textbf{(b):} a stretch preserving the shape of the elements ($\lambda=0$); \textbf{(c):} an area-preserving stretch ($\lambda=10^4$); \textbf{(c):} a trade-off between the shape and the area preservation ($\lambda=1$).}
\label{fig:lambda}
\end{figure}

\paragraph*{Shape-area tradeoff \ $\lambda$}
Our final tests illustrate the influence of the parameter $\lambda$ in Prob.~\eqref{eq:discrete} on the resulting map.
We have computed three free boundary maps of the rest shape (Fig.~\ref{fig:lambda}--a) being stretched.
First we chose $\lambda=0$, that is, only the shape quality term is taken into account in Prob.~\eqref{eq:discrete}.
When we optimize for the angles, the area of the triangles is forced to change, refer to Fig.~\ref{fig:lambda}--b for the resulting map.
Naturally, an area preserving map ($\lambda=10^4$) must deform the elements to satisfy the area constraint (Fig.~\ref{fig:lambda}--c).
Finally, in Fig.~\ref{fig:lambda}--d we show an example with a tradeoff between the area and angles preservation.

\subsection{Limitations}
\label{sec:results:limitations}

While globally performing very well in practice, our method still presents some limitations.
We have two main sources of limitations: theoretical limitations as well as very practical ones related to numerical stability of our resolution scheme.

\paragraph*{Overlaps}
First of all, an inversion-free map does not imply global injectivity.
Fig.~\ref{fig:overlap}--a provides an example of an inversion-free map with two cases of non-injectivity when optimizing for a map with free boundaries:
the map can present global overlaps as well as the boundary can ``wind up'' around boundary vertices, i.e. the total angle of triangles incident to a vertex can be superior to $2\pi$.
Moreover, while being less frequent, similar situations may occur on interior vertices.
Typically this situation happens near constraints causing a local compression in the shape.

Let us illustrate this behavior on a very simplistic mesh consisting of a single fan of 12 triangles.
All vertices are free to move, the target shape is set to be the unit equilateral triangle for all elements.
For this problem Fig.~\ref{fig:covering}--a shows a local minimum, and the Fig.~\ref{fig:covering}--c shows the global minimum respecting perfectly the prescribed total angle of $4\pi$ around the center vertex.
Both are inversion-free maps, but only the map in Fig.~\ref{fig:covering}--a is a globally injective one.
Depending on the initialization and the resolution scheme chosen, we can converge to either minimum.
Note, however, that the center vertex has the winding number 1 in one map and 2 in the other, and thus we can not deform continuously one to the other without inverting some elements.
Note also that the configurations like in the Fig.~\ref{fig:covering}--b present inverted elements and thus can not be generated by our method.

It is possible to avoid all overlaps altogether by embedding our shape to optimize into an outer triangulation, and performing a ``bi-material'' optimization.
In this case, both global overlaps and fold-2-coverings are prohibited by the the fact that the outer material must not have inverted elements (refer to Fig.~\ref{fig:overlap}--b).
The thick prismatic layer in Fig.~\ref{fig:teaser} was generated by a similar procedure: we have generated a very thin layer of triangular prisms around the dragon, and tetrahedralized the exterior bounded by a cube.
After calling the untangling procedure, we have obtained an offset surface with exactly the same mesh connectivity as the original dragon mesh.

While this embedding kind of approach works well for certain applications, for other it may be hard to apply.
We are currently exploring simpler practical ways to fix the winding problem.

\begin{figure}[!t]
\centering
\begin{minipage}{.3\linewidth}
\centering
\vspace{13mm}
\includegraphics[width=\linewidth]{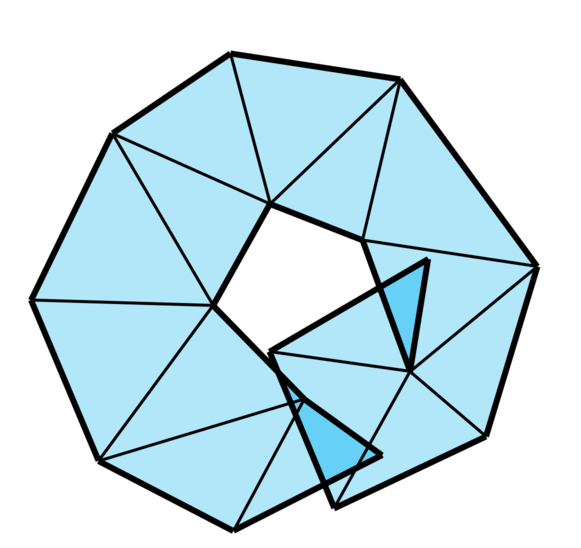}
\textbf{(a)}
\end{minipage}
\begin{minipage}{.65\linewidth}
\centering
\includegraphics[width=\linewidth]{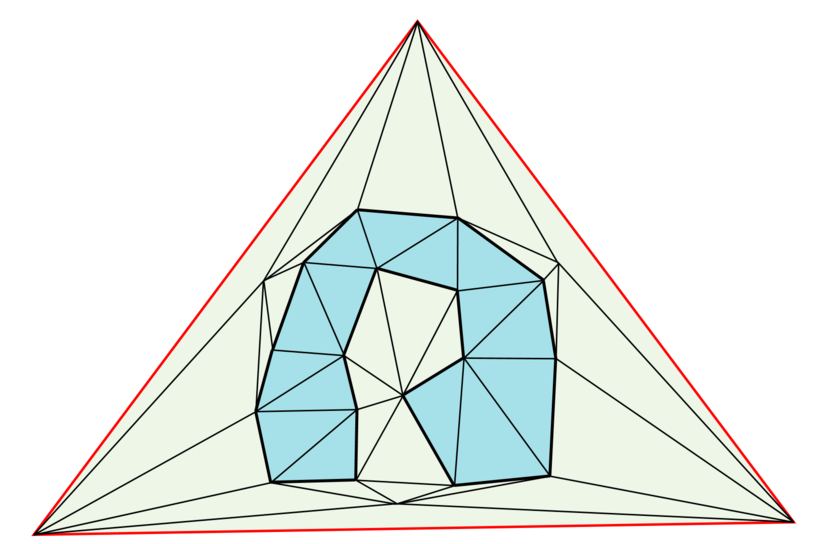}
\textbf{(b)}
\end{minipage}
\caption{Free boundary mapping limitations. \textbf{(a):} this mesh presents two kinds of problems, namely, a global overlap and the mesh wrapped around a boundary vertex. \textbf{(b):} Both problems can be avoided by embedding the mesh into an outer triangulation.}
\label{fig:overlap}
\end{figure}

\begin{figure}[!t]
\centering
\begin{minipage}{.28\linewidth}
\centering
\includegraphics[width=\linewidth]{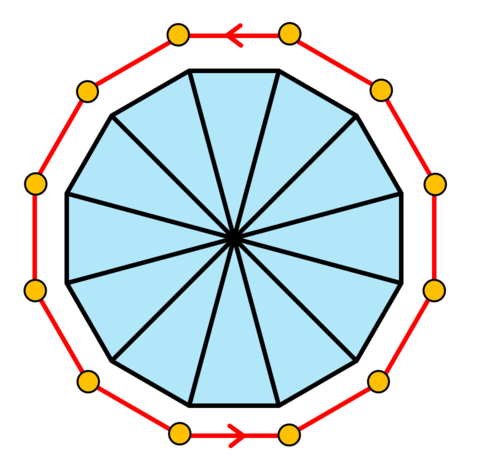}
\textbf{(a)}
\end{minipage}
\begin{minipage}{.3\linewidth}
\centering
\includegraphics[width=\linewidth]{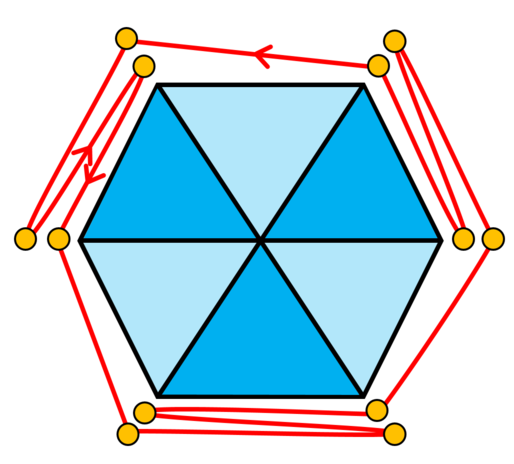}
\textbf{(b)}
\end{minipage}
\begin{minipage}{.3\linewidth}
\centering
\includegraphics[width=\linewidth]{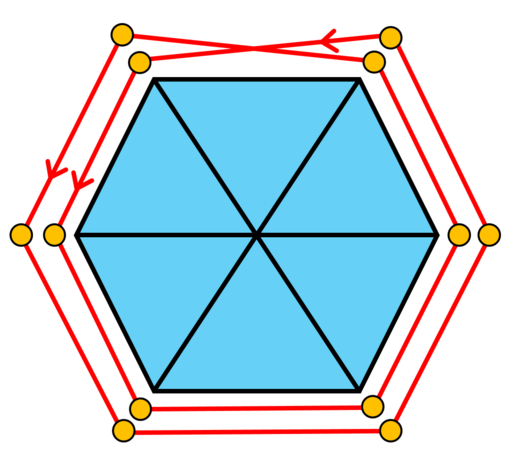}
\textbf{(c)}
\end{minipage}
\caption{Free boundary mapping limitations: three maps of a very simplistic mesh made of 12 triangles. \textbf{(a)} and \textbf{(c)} both are inversion-free maps and thus allowed by our method, whereas the map \textbf{(b)} has inverted elements, and thus is prohibited by our method.}
\label{fig:covering}
\end{figure}

\paragraph*{Numerical challenges}
Even when the problem is well-posed, a robust implementation may present significant difficulties.
As we have said above, the quasi-Newtonian optimization scheme performs well for ``simple'' problems (it passes all the benchmark database!), but may fail for large deformations, where Newton iterations are necessary.
While our modified Hessian matrix is symmetric positive definite, note that for stiff problems the Jacobi preconditioned conjugate gradients can fail and one might need the incomplete Choleski decomposition and beyond.

In practice we have found the method being very robust in 2D settings: we have not encountered a practical test case we were not able to treat with our method.
In 3D, however, it can fail due to the numerical challenges in very anisotropic and highly twisted meshes.


\section{Analysis}
\label{sec:analysis}
This section presents in two parts a rigorous analysis of the penalty method.
First in \S~\ref{sec:hessian} we prove that the modified Hessian matrix $H^+(U, \eps)$ is indeed positive definite, and then in \S~\ref{sec:finite-untangling} we show the origins of Eq.~\eqref{eq:epsilon} for the regularization parameter sequence $\{\eps^k\}_{k=0}^{\dots}$.
Namely, we prove that if the problem has a solution, then an \emph{idealized} minimization method can reach the admissible set $\min \det J>0$ in a finite number of steps.
An immediate consequence of this theorem is that if the problem has a solution, then for some $K<\infty$ the solution $\argmin\limits_U F(U, \eps^K)$ belongs to the admissible set.

\subsection{Modified Hessian matrix}
\label{sec:hessian}
Recall that in our resolution scheme we use the modified Hessian $(d\, \#V) \times (d\, \#V)$ matrix $H^+(U, \varepsilon)$ of the function $F(U, \varepsilon)$ built out of $d\times d$ blocks $H^+_{ij}$
placed on the intersection of $i$-th block row and $j$-th block column.
It is a common practice to add some regularization terms to the Hessian matrix to make it positive definite,
but we propose to modify the finite element (FE) matrix assembly procedure by eliminating some terms potentially leading to an indefinite FE matrix.

To this end, we restrict our attention to a single simplex and we study a function $\phi(J)$ of the Jacobian matrix defined as follows:
\begin{equation}
\label{eq:phi}
\phi(J) := f_\eps(J) + \lambda g_\eps(J) =  \frac{\tr J^\top J}{(\chi(\det J, \varepsilon))^\frac2d} + \lambda \frac{\det^2 J + 1}{\chi(\det J, \varepsilon)}
\end{equation}

Let us denote by $a \in \R^{d^2}$ the (columnwise) flattening of the Jacobian matrix $J$, i.e. the vector composed of the elements of $J$.
We decompose the $d^2 \times d^2$ Hessian matrix of $\phi$ with respect to the Jacobian matrix entries into two parts:
$\frac{\p^2 \phi}{\p a \p a^\top} = M^+ + M^\pm$,
where $M^+$ is a positive definite matrix, and $M^\pm$ can be an indefinite matrix that we neglect.
The matrix $M^\pm$ contains all terms depending on $\chi''$ and second derivatives of $\det J$ with respect to elements of the Jacobian matrix $J$.
Our map is affine on the simplex of interest, therefore its Jacobian matrix $J$ is a linear function of the vertices of the simplex.
Thus, the idea is to compute a positive definite matrix $M^+(J)$, and use the chain rule to get the Hessian matrix with respect to our variables $U$ and assemble the matrix $H^+$.


So, we choose some arbitrary point $J_0$ and we want to show the way to decompose $\frac{\p^2 \phi}{\p a \p a^\top}(J_0)$ into a sum of $M^+(J_0)$ and $M^\pm(J_0)$ with $M^+(J_0)>0$.
To do so, first we write down the first order Taylor expansion $q(D)$ of the function $\chi(D, \varepsilon)$ around some point $D_0 = \det J_0$:
$$
q(D) := \chi(D_0, \eps) + \frac{\p \chi}{\p D}(D_0, \eps) (D - D_0).
$$
Next we define a function $\Phi(a, D)$ as follows:
$$
\Phi(a,D) := \frac{|a|^2}{(q(D))^\frac2d} + \lambda \frac{D^2 + 1}{q(D)}.
$$
Note that $\Phi$ differs a bit from $\phi$: it has one more argument and $\chi$ is replaced by its linearization in the denominator.
While this maneuver might seem obscure, light will be shed very shortly.
$\Phi$ has a major virtue of being convex!
The convexity is easy to prove, refer to Appendix~\ref{app:positive} for a formal proof.

Having built a convex function $\Phi$, it is straightforward to verify that the following decomposition holds:
\begin{equation}
\label{eq:M+M?}
\frac{\p^2 \phi}{\p a \p a^\top}(J_0) = M^+(J_0) + M^\pm(J_0),
\end{equation}
where
\begin{align*}
M^+ &:= \begin{pmatrix} I & \frac{\p D}{\p {a}} \end{pmatrix}  \begin{pmatrix}\frac{\p^2 \Phi}{\p a \p a^\top} & \frac{\p^2 \Phi}{\p a \p D }\\
\frac{\p^2 \Phi}{\p D \p a^\top } & \frac{\p^2 \Phi}{\p D^2} \end{pmatrix} \begin{pmatrix} I & \frac{\p D}{\p {a}} \end{pmatrix}^\top, ~\text{and} \\
M^\pm &:= \frac{\p \Phi}{\p D}\frac{\p^2 D}{\p {a} \p {a}^\top} - \frac{\chi''}{\chi}\left(\frac2d f_\eps + \lambda g_\eps\right)\frac{\p D}{\p a} \frac{\p D}{\p a^\top}.
\end{align*}
The easiest way to check that the equality~\eqref{eq:M+M?} holds is to note that at the point $J_0$ we have $q=\chi$, $q' = \chi'$,
and therefore we have
$$
\frac{\p \phi(a)}{\p {a}} = \frac{\p \Phi(a, D(a))}{\p {a}} + \frac{\p \Phi(a, D(a))}{\p D} \frac{\p D(a)}{\p {a}}.
$$
To calculate the Hessian $\frac{\p^2 \phi}{\p a \p a^\top}(J_0)$, it suffices to differentiate this expression one more time and add the terms in $\chi''$ that were zeroed out by the linearization.


To sum up, in our computations, for each simplex we approximate the Hessian matrix $\frac{\p^2 \phi}{\p a \p a^\top}$ by the $d^2 \times d^2$ matrix $M^+$ and we neglect the term $M^\pm$.
Thanks to the convexity of $\Phi$, it is trivial to verify that for any choice of $J_0$ the matrix $M^+$ is definite positive.
Then we use the chain rule over $M^+$ to get the Hessian matrix with respect to our variables $U$, and we assemble a $(d\, \#V) \times (d\, \#V)$ approximation $H^+$ of the Hessian matrix for the energy function $F(U, \eps)$.
Matrix $H^+$ is positive definite provided that at least $d$ mesh vertices are fixed.
If less than $d$ points are fixed, rigid body transformations are allowed.
The energy is invariant w.r.t rigid body transformations, so when constraints allow for such transformations, matrix $H^+$ becomes positive semidefinite
Note that the leading blocks $H^+_{ii}$ are always positive definite.
Refer to Appendix~\ref{app:gradient-hessian} for further details on the finite element assembly procedure.

\subsection{Choice of $\varepsilon^k$}
\label{sec:finite-untangling}

In this section we provide a strategy for the choice of the regularization parameter $\varepsilon^k$ at each iteration.
Namely, we prove that an idealized minimization algorithm reaches the admissible set $\min \det J>0$ in a finite number of iterations.


\begin{theorem}
\label{th:th}
Let us suppose that the admissible set is not empty, namely there exists a mesh $U^*$ satisfying $F(U^*,0) < + \infty$.
We also suppose that we have a minimization algorithm satisfying one of the following efficiency conditions for some $0<\sigma<1$:
\begin{itemize}
\item \textbf{Either} the essential descent condition holds
\begin{equation}
\label{teorem.cond6}
F(U^{k+1}, \eps^k) \leq (1-\sigma) F(U^{k}, \eps^k),
\end{equation}
\item \textbf{or} the vector $U^{k}$ satisfies the quasi-minimality condition:
\begin{equation}
\label{teorem.cond6-2}
\min\limits_{U} F(U, \eps^k) > (1-\sigma)F(U^{k}, \eps^k).
\end{equation}
\end{itemize}
Then the admissible set is reachable by solving a finite number of minimization problems in $U$ with $\eps^k$ fixed for each problem.
In other words, under a proper choice of the regularization parameter sequence $\varepsilon^k, k=0\dots K$,
we obtain $F(U^{K}, 0) < + \infty$.
\end{theorem}

\begin{proof}
The main idea is to expose an explicit way to build a decreasing sequence $\{\varepsilon^k\}_{k=0}^\infty$ such that the sequence $\{F(U^k, \eps^k)\}_{k=0}^\infty$ is bounded from above.
Then we can prove that the admissible set is reachable in a finite number of steps by a simple \emph{reductio ad absurdum} argument.

First of all, the function $F(U, \eps)$ can be rewritten as follows
\begin{equation}
\label{teorem.eq1}
F(U,\eps) =  \sum\limits_{i} \alpha_i \frac{\psi_i(U)}{\chi(D_i,\eps)},
\end{equation}
where $D_i=D_i(U)$ denotes the Jacobian determinant for $i$-th simplex of the mesh ($D_i=\det J_i$),
and $\alpha_i > 0$ are positive, separated from zero weights assigned to each simplex.
The functions
$$
\psi_i(U, \eps) := \chi(D_i,\eps)^{1 -\frac2d}\tr J_i^\top J_i
+ \lambda (D_i^2 + 1)
$$
defined according to Eq.\eqref{eq:fege} are positive and bounded from below as
$$
\alpha_i \psi_i(U, \eps) \ge \lambda \min_i \alpha_i.
$$
Note also that $\psi_i(U, \eps)$ are increasing functions of $\eps$.

Our goal is to build a decreasing sequence $\{\varepsilon^k\}_{k=0}^\infty$ such that the sequence $\{F(U^k, \eps^k)\}_{k=0}^\infty$ is bounded from above.
We split the construction into two parts: first we suppose that at some iteration $k$ the essential condition \eqref{teorem.cond6} is satisfied, and then we explore the case~\eqref{teorem.cond6-2}.

\textit{Suppose that the condition \eqref{teorem.cond6} holds at iteration $k$.}
In order to guarantee that the function does not increase, it suffices to establish the following inequality:
\begin{equation}
\label{f_bounded}
(1 - \sigma) F(U^{k+1}, \eps^{k+1}) \leq F(U^{k+1},\eps^{k}).
\end{equation}


By noting that
$\psi_i(U^{k+1}, \eps^{k+1}) \leq \psi_i(U^{k+1}, \eps^{k})$, Ineq.~\eqref{f_bounded} is implied if the following condition holds:
\begin{equation}
\label{suff_cond}
\forall i: \qquad
(1-\sigma)\chi(D_i^{k+1},\eps^{k})  \leq \chi(D_i^{k+1},\eps^{k+1}) 
\end{equation}
where
$D_i^{k+1} := D_i(U^{k+1})$ denotes the Jacobian determinant of simplex $i$ at iteration $k+1$.

Let us show a constructive way to build $\eps^{k+1}$ such that Ineq.~\eqref{suff_cond} is satisfied.
To do so, first note that $\chi$ is convex with respect to the second argument, hence
$$
\chi(D_i^{k+1},\eps^{k+1})-\chi(D_i^{k+1},\eps^k) \ge (\eps^{k+1}-\eps^k) \frac{\partial}{\partial\eps} \chi(D_i^{k+1},\eps^k).
$$
Thus, since $\chi$ and $\frac{\partial}{\partial\eps} \chi(D,\eps)$ are both positive, Ineq.~\eqref{suff_cond} holds provided that
\begin{equation}
\label{suff2cond}
\forall i: \qquad \eps^{k+1} \ge \eps^k - \sigma  \frac{\chi(D_i^{k+1},\eps^k)}{\frac{\partial}{\partial\eps}\chi(D_i^{k+1},\eps^k)}
\end{equation}

In its turn, since $\frac{\chi(D,\eps)}{\frac{\partial}{\partial\eps}\chi(D,\eps)}$ is an increasing function of $D$,
condition~\eqref{suff2cond} can be simplified to
$$
\eps^{k+1} \ge \eps^k - \sigma  \frac{\chi(D_-^{k+1},\eps^k)}{\frac{\partial}{\partial\eps}\chi(D_-^{k+1},\eps^k)},
$$
where $D^{k+1}_- := \min\limits_{i}~D_i^{k+1}$ is the minimum value of the Jacobian determinant over all simplices.

Hence, if $U^{k+1}$ is an approximate solution of the minimization problem $\argmin\limits_U F(U, \eps^k)$ with fixed parameter $\eps^k$,
we can use the following update rule for $\eps^{k+1}$:
\begin{equation}
\label{teorem.eq.eps}
\eps^{k+1} =  \left\{ 
\begin{array}{lcl} 
\eps^{k} - \sigma\frac{\chi(D_-^{k+1},\eps^k)}{\frac{\partial}{\partial\eps}\chi(D_-^{k+1},\eps^k)},
&\mbox{ if }& D_-^{k+1} < 0, \\ 
\eps^{k} - \sigma\frac{\chi(0,\eps^k)}{\frac{\partial}{\partial\eps}\chi(0,\eps^k)},
&\mbox{ if }& D_-^{k+1} \ge 0, \\
\end{array} 
\right.
\end{equation}
This update rule guarantees that Ineq.~\eqref{f_bounded} is satisfied; coupled with the assumption \eqref{teorem.cond6} of the theorem,
this implies the required non-growth property of the function values sequence:
\begin{equation}
\label{f-non-increase}
F(U^{k+1}, \eps^{k+1}) \leq F(U^{k}, \eps^{k}). 
\end{equation}

\textit{Consider now the case where condition~\eqref{teorem.cond6-2} holds at iteration $k$.}
Note that condition~\eqref{teorem.cond6-2} essentially means that our current solution $U^k$ is very close to the global minimum of $F(U, \eps^k)$,
and thus Ineq.~\eqref{teorem.cond6} cannot be satisfied.
Nevertheless, we can use the same update rule \eqref{teorem.eq.eps} for computation of $\eps^{k+1}$.
Indeed, with this choice we have
\begin{align*}
F(U^{k+1}, \eps^{k+1})  & \leq \frac1{(1-\sigma)} F(U^{k+1}, \eps^{k}) \leq \frac1{(1-\sigma)}  F(U^{k}, \eps^{k}) < \\
& < \frac1{(1-\sigma)^2} \min\limits_{U} F(U, \eps^k) < \frac1{(1-\sigma)^2} \min\limits_{U} F(U, 0).
\end{align*}
Here the last inequality provides a global bound on the function values sequence, and it is 
based on the observation $\frac{\partial}{\partial\eps} \chi(D,\eps) > 0$.

\vspace{2mm}

To sum up, we have shown a way to build a sequence $\{\varepsilon^k\}_{k=0}^\infty$ such that the sequence $\{F(U^k, \eps^k)\}_{k=0}^\infty$ is bounded from above.
Now let us prove that the update rule~\eqref{teorem.eq.eps} allows to reach the admissible set in a finite number of steps.
To do so, we use a \emph{reductio ad absurdum} argument.

Suppose that the admissible set is never reached for an infinite decreasing sequence $\{\eps^{k}\}_{k=0}^{\infty}$ built using the update rule~\eqref{teorem.eq.eps},
i.e.  $D_-^{k+1} < 0 ~ \forall k\geq 0$.

Then, using the facts that $\frac{\partial}{\partial\eps} \chi(D,\eps) \le \frac12$ and $\{\varepsilon^k\}_{k=0}^\infty$ form a decreasing sequence of non-negative values,
we can rewrite the update rule~\eqref{teorem.eq.eps} for some $K>0$:
$$
\eps^0-\eps^K  \ge 2\sigma\sum_{k=0}^{K-1}\chi(D_-^{k+1},\eps^{k}) \ge 2 \sigma K\min_{0\le k <K}\chi(D_-^{k+1},\eps^{k}), 
$$
with an immediate consequence that for an arbitrarily large $K$ we have the following inequality:
$$
\max_{0\le k <K}\frac{1}{\chi(D_-^{k+1},\eps^{k})} \ge \frac{2 \sigma K}{\eps^0}.
$$
Since all terms $\alpha_i \psi_i(U)$ in (\ref{teorem.eq1}) are bounded from below, 
the resulting estimate contradicts the boundedness of $F(U^k, {\eps^k})$, thus concluding our proof.

\end{proof}

\begin{remark}
An important corollary of Th.~\ref{th:th} is that, provided that the admissible set is not empty,
there exists an iteration $K<\infty$ such that the global minimum of the function $F(U, \eps^K)$ belongs to the admissible set.
The proof is rather obvious: suppose we have an idealized minimizer such that $U^{k+1} = \argmin\limits_U F(U, \eps^k)$.
This minimizer always satisfies the conditions of Th.~\ref{th:th}, therefore it can untangle the mesh in a finite number of steps.
\end{remark}

\begin{remark}
For our choice of the regularization function $\chi$, the update rule~\eqref{teorem.eq.eps} can be instatiated as follows:
\begin{equation}
\label{eps.update}
\eps^{k+1} =  \left\{ 
\begin{array}{lcl} 
\left(1 - \frac{\sigma\sqrt{(D_-^{k+1})^2+(\eps^k)^2}}
{|D_-^{k+1}|+\sqrt{(D_-^{k+1})^2+(\eps^k)^2}}\right)\eps^{k}, 
&\mbox{ if }& D_-^{k+1} < 0, \\ 
&&\\
(1-\sigma)\eps^{k}, 
&\mbox{ if }& D_-^{k+1} \ge 0, \\
\end{array} 
\right.
\end{equation}

In practice the global estimate $\sigma$ is not known in advance, and the optimization routine may be far from the ideal.
For each minimization step we compute the local descent coefficient $\sigma^k$:
\[
\sigma^k := 1 -  \frac{F(U^{k+1}, \eps^k)}{F(U^{k}, \eps^k)}.
\]
When $\sigma^k \geq \sigma$ one can use the update rule \eqref{eps.update} using the local value $\sigma^k$ guaranteeing that Ineq.~\eqref{f-non-increase} holds. 
In the case $\sigma^k < \sigma$ one should check that condition \eqref{teorem.cond6-2} holds for prescribed $\sigma$. If positive, we can assign $\sigma^k = \sigma$ and use update rule \eqref{eps.update}.
If one cannot assure \eqref{teorem.cond6-2}, it means that minimization procedure for $\eps^k$ failed and theorem cannot be applied (it does not mean that Alg.~\ref{alg:lbfgs} will not reach the admissible set!).
In numerical experiments we use value $\sigma = \frac1{10}$.
\end{remark}


\section{Conclusion}
\label{sec:conclusion}

Producing maps without reverted elements is a major challenge in geometry processing. Inspired by untangling solutions in computational physics, our solution outperforms the state of the art both in terms of robustness and distortion. It is easy to use since we provide a simple implementation that is free of commercial product dependency (compiler, library, etc.).
Moreover, the energy is estimated independently on each triangle / tetrahedra, making it a good candidate to be adapted to more difficult settings including free boundary and global parameterization.


\bibliographystyle{ACM-Reference-Format}
\bibliography{instantangle.bib}

\appendix

\section{Comprehensive design formulæ}
\label{app:gradient-hessian}

\begin{figure}[!t]
\centering
\includegraphics[width=.7\linewidth]{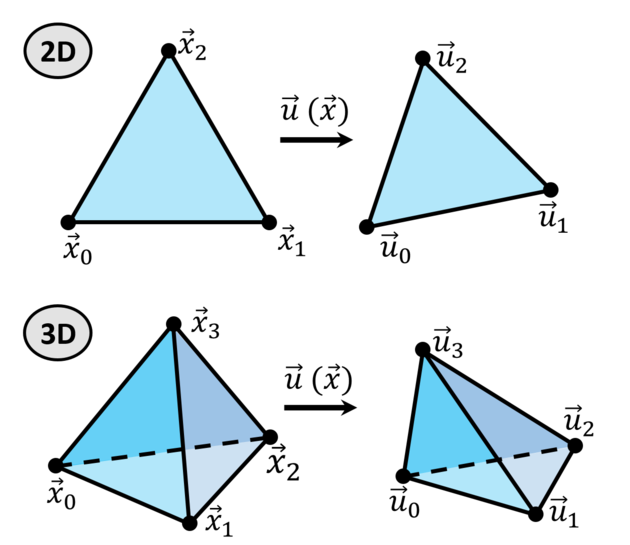}
\caption{On each simplex the map $\vec u(\vec x)$ is affine and is entirely defined by the position of the vertices of the domain simplex $\{\vec{x}_i\}$ and its image $\{\vec u_i\}$.}
\label{fig:simplex}
\end{figure}

Given a map $\vec{u}$, let us denote by $\vec{a}_i,~~ i = 1,2~(,3)$ the tangent basis, i.e. vectors forming the columns of the Jacobian matrix $J$.
For example, in 2D we have $\vec{a}_1:=\begin{pmatrix}\frac{\p u}{\p x} & \frac{\p v}{\p x}\end{pmatrix}^\top$ and $\vec{a}_2:=\begin{pmatrix}\frac{\p u}{\p y} & \frac{\p v}{\p y}\end{pmatrix}^\top$.
Let us denote by $\vec{b}_i$ the dual basis, i.e. vectors chosen in the way that $\vec{a}_i^\top \vec{b}_j =  \delta_{ij} \det J$ for all indices $i,j$.
In particular, for the 2D settings the dual basis can be written as $\vec{b}_1 := \begin{pmatrix}\frac{\p v}{\p y} & -\frac{\p u}{\p y}\end{pmatrix}^\top$
and $\vec{b}_2 := \begin{pmatrix} -\frac{\p v}{\p x} & \frac{\p u}{\p x}\end{pmatrix}^\top$.
In the 3D case $\vec{b}_k = \vec{a}_i \times \vec{a}_j$, where $i,j,k$ is cyclic permutation from $1,2,3$.
It is a handy choice of variables, in particular, $\tr J^\top J = \sum_i |\vec{a}_i|^2$ and $\frac{\p \det J}{\p \vec{a}_i} = \vec{b}_i$.
For further simplification of notations we will use $\chi$ for $\chi(D, \eps)$, $\chi'$ for
$\frac{\p \chi(D, \eps)}{\p D}$ and
$a^\top = (\vec{a}_1^\top \dots \vec{a}_d^\top), \ \  b^\top = (\vec{b}_1^\top \dots \vec{b}_d^\top)$.

\subsection{Gradient}

In order to derive expressions for the gradient and the Hessian matrix of $F$, we write down explicitly
the Jacobian matrix $J$ for the affine map of a simplex $T$ with vertices $\vec{u}_0, \vec{u}_1, \dots, \vec{u}_d$:
\begin{align*}
J & = (\vec{a}_1 \dots \vec{a}_d) = (\vec{u}_1 - \vec{u}_0 \ \vec{u}_2 - \vec{u}_0 \dots \vec{u}_d - \vec{u}_0)\, S^{-1} =\\
  & = (\vec{u}_0 \dots \vec{u}_d) Z,
\end{align*}
where
\[
S := (\vec{x}_1 - \vec{x}_0 \ \vec{x}_2 - \vec{x}_0 \dots \vec{x}_d - \vec{x}_0), \ \ \det S > 0
\]
is the shape matrix, $\vec{x}_i$ are vertices of ``ideal'' or ``target'' shape for the image of the simplex $T$,
and $Z$ is a $(d+1)\times d$ matrix defined as
$$
Z:=\{z_{ij}\}:=\begin{pmatrix}-1 & \dots & -1\\ & I &  \end{pmatrix} S^{-1} 
$$

Since the Jacobian matrix is a linear function of $\vec{u}_i$, we have
\[
\frac{\p \vec{a}_i}{\p \vec{u}_j^\top} = z_{ji} I, \quad i = 1, \dots, d,  \ j = 0, \dots, d.
\]
The additive contribution to gradient of $F$ from the simplex $T$ can be written using correspondence of local indices $0-d$ and global indices $g_0 - g_d$ in the list of vertices:
\begin{align*}
(\nabla F)_{g_j} \mathrel{+}=& \frac{\det S}{d!} \sum\limits_{i=1}^d  \frac{\p \vec{a}_i^\top}{\p \vec{u}_j} \frac{\p \phi}{\p \vec{a}_i} = \\
        =  &\frac{\det S}{d!} \sum\limits_{i=1}^d z_{ji} \frac{\p \phi}{\p \vec{a}_i}, \quad \ j = 0, \dots,d,
\end{align*}
where function $\phi(J)$ is defined in \S~\ref{sec:hessian}.
Let us provide an explicit expression for $\frac{\p \phi}{\p \vec{a}_i}$:
$$ \frac{\p \phi}{\p \vec{a}_i} = 
\frac{2}{\chi^\frac2d} \vec{a}_i -\frac{1}{\chi}\left(\frac2d f_\varepsilon\chi' - 2\lambda \det J + \lambda g_\varepsilon \chi'\right)\vec{b}_i
$$

\subsection{Hessian}

The blocks of the nonnegative definite part of Hessian matrix of $F$ can be updated using the following general formula
\[
H^+_{g_j g_i} \mathrel{+}=  \frac{\det S}{d!} \sum_{m,l} \frac{\p \vec{a}_m^\top}{\p \vec{u}_j} M^+_{ml} \frac{\p \vec{a}_l}{\p \vec{u}_i^\top},
\]
where $M^+_{ml}$ denotes a $d \times d$ block of $d^2\times d^2$ positive definite matrix $M^+$ defined in Eq.~\eqref{eq:M+M?}.
Let us provide an explicit expression for the matrix:
$$
M^+ = \begin{pmatrix} I & b \end{pmatrix}  \begin{pmatrix}\frac{\p^2 \Phi}{\p a \p a^\top} & \frac{\p^2 \Phi}{\p a \p D }\\
\frac{\p^2 \Phi}{\p D \p a^\top } & \frac{\p^2 \Phi}{\p D^2} \end{pmatrix} \begin{pmatrix} I & b \end{pmatrix}^\top, ~\text{where} 
$$
\begin{align*}
\frac{\p^2 \Phi}{\p a \p a^\top} &= \frac2{\chi^\frac2d} I \\
\frac{\p^2 \Phi}{\p D^2} &= \frac2d \left(1 + \frac2d\right) |a|^2 \frac{{\chi'}^2}{\chi^{2 + \frac2d}} + \lambda \left(\frac2{\chi} - 4D \frac{\chi'}{\chi^2} + 2(1 + D^2) \frac{{\chi'}^2}{\chi^3} \right)\\
\frac{\p^2 \Phi}{\p a \p D} &= -\frac4d \frac{\chi'}{\chi^{1 + \frac2d}} a.
\end{align*}

Obviously the leading $d \times d$ blocks of the matrix $H^+$ are strictly positive definite and can be used to build  Newton-type minimization algorithm.

\section{Convexity of $\Phi$}
\label{app:positive}
Recall that the function $\Phi$ is defined as follows (\S~\ref{sec:hessian}):
\[
\Phi(a,D) := \frac{|a|^2}{q^\frac2d} + \lambda \frac{D^2 + 1}{q},
\]
where $a \in \R^{d^2}$ is the (columnwise) flattening of the Jacobian matrix $J$, i.e. the vector composed of the elements of $J$.

\vfill\pagebreak

\begin{lemma}
$\nabla \nabla^\top \Phi > 0$
\end{lemma}

\begin{proof}
It is straightforward to see that the  $(d^2+1)\times (d^2+1)$ Hessian matrix of $\Phi$ can be written in the $2 \times 2$ block representation:
$$
\nabla \nabla^\top \Phi = P + \lambda Q,
$$
where
\begin{align*}
P &:= \begin{pmatrix} \frac2{q^\frac2d} I & - \frac4d \frac{q'a}{q^{1 + \frac2d}}\\ - \frac4d \frac{q'a^\top}{q^{1 + \frac2d}} & \frac2d (1 + \frac2d )\frac{|a|^2q'^2}{q^{2 + \frac2d}}\end{pmatrix} ~~\text{and}\\
Q &:= \begin{pmatrix} 0 & 0 \\ 0 & \frac2{q} - 4D \frac{q'}{q^2} + 2(1 + D^2) \frac{{q'}^2}{q^3}\end{pmatrix}.
\end{align*}
It is trivial to verify that $Q\geq 0$, since $Q_{22}$ is a strictly positive quadratic function of argument $D$.
Since the leading blocks of the matrix $P$ are positive definite and the Schur complement
\[
P_{22} - P_{21} P_{11}^{-1} P_{12} = \frac{|a|^2}{q^{2 + \frac2d}} \frac2d \left(1 - \frac2d \right) \geq 0
\]
is nonnegative definite, overall convexity of $\Phi$ is established.
\end{proof}

\begin{remark}
\label{remark:polyconvexity}
Note that we have just proved the convexity of the function $\Phi$.
As an immediate consequence we obtain polyconvexity of the functional~\eqref{eq:winslow}, because it is a particular case of our functional~\eqref{eq:continuous} with $\chi=q$.
\end{remark}




\section*{Supplementary material}

Listing 1. A complete L-BFGS-based mesh untangling example

```
1  from mesh import Mesh
2  import numpy as np
3  from scipy.optimize import fmin_l_bfgs_b
4
5  mesh = Mesh() # generate a test quad mesh
6  n = mesh.nverts
7  Q = [ np.matrix('-1,-1;1,0;0,0;0,1'), np.matrix('-1,0;1,-1;0,1;0,0'),  # quadratures for
8        np.matrix('0,0;0,-1;1,1;-1,0'), np.matrix('0,-1;0,0;1,0;-1,1') ] # every quad corner
9
10 def jacobian(U, qc, quad): # evaluate the Jacobian matrix at the given quadrature point
11     return np.matrix([[U[quad[0]  ], U[quad[1]  ], U[quad[2]  ], U[quad[3]  ]],
12                       [U[quad[0]+n], U[quad[1]+n], U[quad[2]+n], U[quad[3]+n]]]) * Q[qc]
13
14 for iter in range(10): # outer L-BFGS loop
15     mindet = min( [ np.linalg.det( jacobian(mesh.x, qc, quad) ) for quad in mesh.quads for qc in range(4) ] )
16     eps = np.sqrt(1e-6**2 + .04*min(mindet, 0)**2); # the regularization parameter ε
17
18     def energy(U): # compute the energy and its gradient for the map u⃗
19         F,G = 0, np.zeros(2*n)
20         for quad in mesh.quads: # sum over all quads
21             for qc in range(4): # evaluate the Jacobian matrix for every quad corner
22                 J = jacobian(U, qc, quad)
23                 det = np.linalg.det(J)
24                 chi = det/2 + np.sqrt(eps**2 + det**2)/2    # the penalty function χ(ε,det(J))
25                 chip = .5 + det/(2*np.sqrt(eps**2 + det**2)) # its derivative χ'(ε,det(J))
26                 f = np.trace(np.transpose(J)*J)/chi # quad corner shape quality
27                 F += f
28                 dfdj = (2*J - np.matrix([[J[1,1],-J[1,0]],[-J[0,1],J[0,0]]])*f*chip)/chi # ∂f/∂J⃗_qc: derivative w.r.t the Jacobian
29                 dfdu = Q[qc] * np.transpose(dfdj) # chain rule for the real variables
30                 for i,v in enumerate(quad):
31                     if (mesh.boundary[v]): continue # the boundary verts are locked
32                     G[v  ] += dfdu[i,0]
33                     G[v+n] += dfdu[i,1]
34         return F,G
35     mesh.x = fmin_l_bfgs_b(energy, mesh.x)[0] # inner L-BFGS loop
36 print(mesh) # print wavefront .obj file
```

Listing 2. A simplistic quad mesh class

```
1  import numpy as np
2
3  class Mesh():
4      def __init__(self): # generate the test problem: a regular 2d grid with upper half shifted
5          n = 8
6          self.x = [   i/n + int(j>=n//2)*3/5 for j in range(n) for i in range(n) ] + \
7                   [ 2*j/n - int(j>=n//2)*3/5 for j in range(n) for i in range(n) ]      # 2D geometry
8          self.quads = [ [i+j*n, i+1+j*n, i+1+(j+1)*n, i+(j+1)*n] for j in range(n-1) for i in range(n-1) ] # connectivity
9          self.boundary = [ i==0 or i==n-1 or j==0 or j==n-1  for j in range(n) for i in range(n) ] # vertex boundary flags
10
11     @property
12     def nverts(self):
13         return len(self.x)//2
14
15     def __str__(self): # wavefront .obj output
16         ret = ""
17         for v in range(self.nverts):
18             ret = ret + ("v %f %f 0\n" % (self.x[v], self.x[v+self.nverts]))
19         for f in self.quads:
20             ret = ret + ("f %d %d %d %d\n" % (f[0]+1, f[1]+1, f[2]+1, f[3]+1))
21         return ret
```



\end{document}